\documentclass[a4paper,english]{article}
\usepackage{microtype}
\usepackage{graphicx}
\usepackage[dvipsnames]{xcolor}
\usepackage{stmaryrd}
\usepackage{amsmath}
\usepackage{amsthm}
\usepackage{amsfonts}
\usepackage{xspace}
\usepackage{enumitem}
\usepackage{authblk}
\usepackage{hyperref}

\graphicspath{{./eurocg21/figures/},{./figures/}}

\newcommand{\var}[1]{\ensuremath{\left\llbracket \normalfont \texttt{#1}\right\rrbracket}}
\newcommand{\seg}{\ensuremath{\textrm{seg}}\xspace}

\newcommand{\pointvisibilitywidth}{point visibility width\xspace}
\newcommand{\Pointvisibilitywidth}{Point visibility width\xspace}

\newcommand{\chordvisibilitywidth}{chord visibility width\xspace}
\newcommand{\Chordvisibilitywidth}{Chord visibility width\xspace}

\newcommand{\iteratedComb}{Iterated Comb\xspace}

\newtheorem{theorem}{Theorem}
\newtheorem{lemma}[theorem]{Lemma}
\newtheorem{claim}{Claim}

\newcommand{\till}[1]{\noindent \textcolor{Bittersweet}{\textbf{Till:}  #1 }    }

\title{Local Complexity of Polygons}

\author[1]{Fabian Klute\thanks{Supported by the Netherlands Organisation for Scientific Research (NWO) under project no. 612.001.651.}}
\author[2]{Meghana M.~Reddy\thanks{Supported by the Swiss National Science Foundation within the collaborative DACH project \emph{Arrangements and Drawings} as SNSF Project 200021E-171681.}\footnote{The second author's full last name consists of two words and is \emph{Mallik Reddy}. However, she consistently refers to herself with the first word of her last name being abbreviated.}}
\author[1]{Tillmann Miltzow\thanks{Supported by the NWO Veni project EAGER.}}
\affil[1]{ETH Z\"{u}rich, Department of Computer Science}
\affil[2]{Utrecht University, Information and Computing Science Department}

\begin{document}
\maketitle

\begin{abstract}
    Many problems in Discrete and Computational Geometry deal with simple polygons or polygonal regions.
    Many algorithms and data-structures perform considerably faster, if the underlying polygonal region has low local complexity.
    One obstacle to make this intuition rigorous, is the lack of a formal definition of local complexity.
    Here, we give two possible definitions and show how they are related in a combinatorial sense.
    We say that a polygon~$P$ has \emph{\pointvisibilitywidth} $w=\var{pvw}$,
    if there is no point $q\in P$ that sees more than~$w$ reflex vertices.
    We say that a polygon~$P$ has \emph{\chordvisibilitywidth} $w=\var{cvw}$,
    if there is no chord $c=\seg(a,b)\subset P$ that sees more than~w reflex vertices.
    We show that 
    \[\var{cvw} \leq \var{pvw}^{O(\var{pvw})},\]
    for any simple polygon.
    Furthermore, we show that there exists a simple polygon with
    \[\var{cvw} \geq 2^{\Omega(\var{pvw})}.\]

\end{abstract}

\section{Introduction}
In Discrete and Computational Geometry we study many problems with respect to the input size~$n$ and other natural parameters.
One famous example is the computation of the convex hull of a set of points in the plane. While $\Theta(n\log n)$ time is worst case possible, this can be improved to $\Theta(n\log h)$, where $h$ is the number of vertices on the convex hull~\cite{kirkpatrick1986ultimate}.
Here, the number of vertices on the convex hull is a natural parameter to study this problem.
We also say sometimes that the algorithm is output-sensitive.
Another famous example, is the spread~$\Delta$ of a set of points in the plane. That is the ratio between the largest and the smallest distance, between any two points. 
Efrat and Har-Peled were the first to find an approximation algorithm for the art gallery problem under the assumption that the underlying set of vertices has bounded spread~\cite{EFRAT2006238}.
A third example is the number of reflex vertices of a polygon. 
This parameter gave raise to some FPT algorithms for the art gallery problem~\cite{agrawal2020parameterized}.

In this work, we introduce two new parameters that are meant to capture rigorously the idea of local complexity. 
Consider the polygons shown in Figure~\ref{fig:localComplexity},
most researchers would probably agree that the polygon on the left 
has lower local complexity than
the polygon on the right. 
Yet, it is not straightforward how to 
define this rigorously in a mathematical sense.

\begin{figure}
    \centering
    \includegraphics{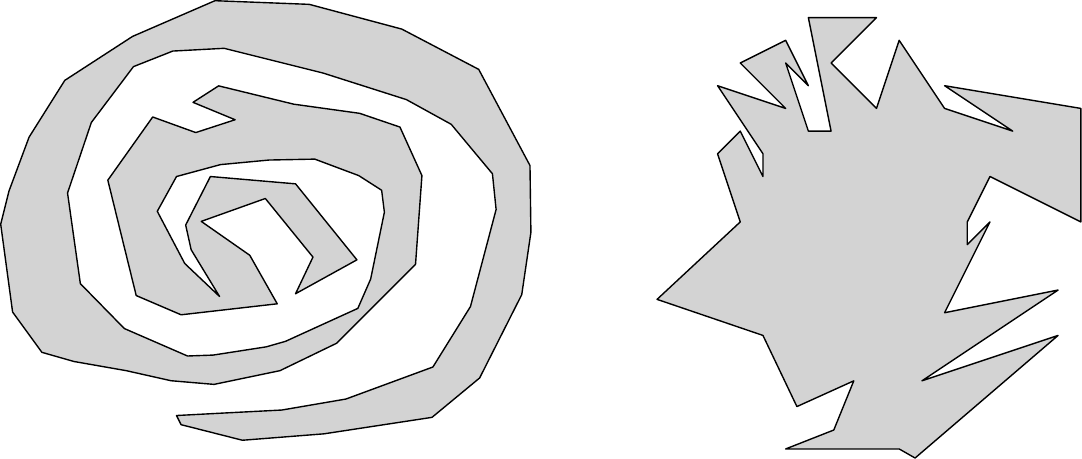}
    \caption{The polygon on the left has intuitively lower local complexity than on the right.}
    \label{fig:localComplexity}
\end{figure}

 Here, we give two possible definitions and show how they are related in a combinatorial sense.
    We say that a polygon~$P$ has \emph{\pointvisibilitywidth} $w=\var{pvw}$,
    if $w$ is the smallest number such that there is no point $q\in P$ that sees more than~$w$ reflex vertices.
    We say that a polygon~$P$ has \emph{\chordvisibilitywidth} $w=\var{cvw}$,
    if $w$ is the smallest number such that there is no chord $c=\seg(a,b)\subset P$ that sees more than~$w$ reflex vertices.

    We show the following theorem.
    
    \begin{theorem} 
    \label{thm:main}
    For every polygon with \chordvisibilitywidth~\var{cvw} and \pointvisibilitywidth ~\var{pvw}, it holds that
     \[\var{pvw}\leq \var{cvw} \leq \var{pvw}^{O(\var{pvw})}.\]
     Moreover, there are polygons such that 
     \[ \var{cvw} \geq 2^{\Omega(\var{pvw})}.\] 
    \end{theorem}

    Note that Hengeveld and Miltzow already defined the notion of \chordvisibilitywidth in a very similar way~\cite{SimonARTGALLERY}.
    Specifically, they showed that the art gallery problem admits an FPT algorithm with respect to \chordvisibilitywidth.
    For a parameter to be interesting to study, we usually have three
    criteria.
\begin{itemize}[leftmargin = 2cm, itemsep = 0pt]
    \item[naturalness:] Although there is no definition of what it means to be mathematically natural, many researchers seem to have a common understanding of this notion.
    \item[relevance:] The parameter is at least for some fraction of instances reasonably low.
    \item[profitable:] Using the parameter, we should be able to 
    design better algorithms and prove useful run time upper bounds.
\end{itemize}
    
    We believe that both parameters are mathematically natural.
    Theorem~\ref{thm:main} indicates that the \chordvisibilitywidth can be exponentially larger than the \pointvisibilitywidth.
    Thus we would expect that \chordvisibilitywidth is potentially more profitable. 
    We would expect that both parameters are equally relevant as the example that we give is fairly contrived.
    The remainder of this paper is dedicated to proving Theorem~\ref{thm:main}.

\section{\Chordvisibilitywidth vs \Pointvisibilitywidth}
We prove Theorem~\ref{thm:main} in two parts.
First, we show the second half of the theorem in Section~\ref{sec:lowerbound} 
by constructing a polygon for which it holds that 
$\var{cvw} \geq 2^{\Omega(\var{pvw})}$.
Second, we show the first half of Theorem~\ref{thm:main} in Section~\ref{sec:upperbound}
by analysing how the reflex vertices visible from a chord in a simple polygon $P$ 
restrict each others vision
and relating this to the \pointvisibilitywidth of the polygon.

\subsection{Lower bound}
\label{sec:lowerbound}

\begin{figure}[h]
    \centering
    \includegraphics[page =3]{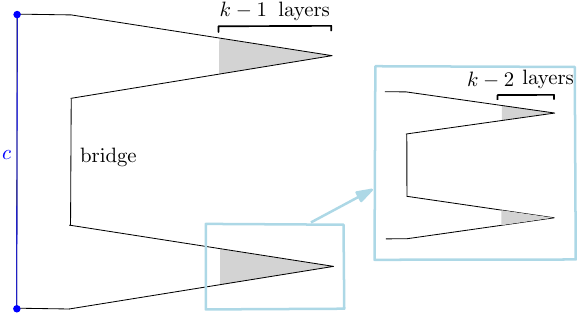}
    \caption{Construction of the \iteratedComb. 
    }
    \label{fig:polygon}
\end{figure}
We construct a
polygon $P$, called the \emph{\iteratedComb}, see Figure~\ref{fig:polygon}.
In the following, let $k \in \mathbb{N}$. The \iteratedComb consists of $k$ \emph{layers}, each layer consists of two \emph{spikes} and each spike further splits into two more spikes in the subsequent layer.
Observe, that the entire polygon is visible from the chord connecting the two left-most points of the polygon. The distance between consecutive spikes in a layer, referred to as the \emph{bridge}, is adjusted such that if at least one vertex in the interior of a spike is visible from a point $p$ on $c$, then $p$ cannot see any interior vertex of any other spike. This property is achieved by stretching the bridges vertically. More specifically, for $ 1 < i \leq k$, the length of the bridge of the $i^{th}$ layer is increased such that the property holds for layer $i$ and then the bridge of the previous layer is adjusted accordingly. By iteratively stretching the bridges from the last layer to the first layer, it can be ensured that the property holds for every layer.
This property is illustrated in Figure~\ref{fig:distance_bw_spikes} for $k=2$. In the first layer, the point $p$ sees an interior vertex of the first spike and no interior vertex of the second spike. Similarly, in the second layer point $p$ sees an interior vertex of the second spike and no interior vertex of the first spike.

\begin{figure}[h]
    \centering
    \includegraphics[page = 3]{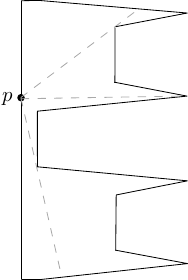}
    \caption{Point $p$ sees interior points of at most one spike of any layer}
    \label{fig:distance_bw_spikes}
\end{figure}

\paragraph*{\Chordvisibilitywidth of the \iteratedComb} 
Clearly, the chord which sees the highest number of reflex vertices is the chord defined by the two left-most vertices. Let this chord be $c$. The number of reflex vertices of the first layer visible from $c$ is two. Similarly, the number of reflex vertices of the $i^{th}$ layer visible from $c$ is $2^i$. Summing up over all $k$ layers, the number of reflex vertices visible from $c$ is $\Theta(2^{k+1})$, and hence $\var{cvw}=\Theta(2^{k+1})$.

\paragraph*{\Pointvisibilitywidth of the \iteratedComb}
\begin{claim}
Chord $c$ contains at least one of the points in $P$ which see the highest number of reflex vertices of $P$.
\end{claim}
\begin{proof}
Let $q$ be a point in polygon $P$ which sees the highest number of reflex vertices of $P$. Let $p$ be a point on chord $c$ which has the same $y$-coordinate as $q$. Assume $p \neq q$. Let $r$ be a reflex vertex visible from $q$. Since $P$ is monotone with respect to y-axis, the triangle $pqr$ must be empty. This implies that $r$ is visible from $p$ as well. 
Hence, the point $p$ also sees the highest number of reflex vertices in $P$ since $p$ sees at least as much as $q$. 
Refer to Figure~\ref{fig:point_on_chord_sees_more} for an illustration.
\end{proof}

\begin{figure}[h]
    \centering
    \includegraphics[page=2]{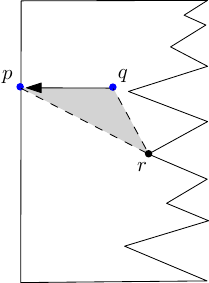}
    \caption{Point $p$ sees all the reflex vertices visible from $q$}
    \label{fig:point_on_chord_sees_more}
\end{figure}

Without loss of generality, assume the point with highest visibility is the topmost point on $c$, denoted by $p$. Both the reflex vertices in layer one are visible from $p$. In each subsequent layer, $p$ can see the reflex vertices that are in the interior of the first spike, which is two reflex vertices, $p$ cannot see any of the other reflex vertices in the other spikes by construction. Summing it up, we can conclude that $2k$ reflex vertices are visible from $p$, and thus $\var{pvw}=2k$.
Hence the \iteratedComb has $\var{cvw} \geq 2^{\Omega(\var{pvw})}$.

\subsection{Upper bound}
\label{sec:upperbound}
Next, we show that 
we can upper bound the \chordvisibilitywidth in terms of the \pointvisibilitywidth.

To this end, we prove the following lemma.
\begin{lemma}
    \label{lem:Point-Chord-UpperBound}
    For every simple polygon, it holds that
    \[\var{cvw} \leq \var{pvw}^{O(\var{pvw})}.\]
\end{lemma}
The rest of this paragraph is dedicated to the proof of Lemma~\ref{lem:Point-Chord-UpperBound}.

For that purpose assume, we are given  a simple polygon $P$ together with a chord $s \subset P$.
Furthermore, we assume that no point in $P$ sees more than $k = \var{pvw}$ reflex vertices of $P$.
Let us denote by $R$ the set of all reflex vertices that see at least one point of
$s  = \seg(a,b)$. 
Furthermore, we also include the two endpoints of $s$ in the set $R$.
As $P$ is a simple polygon it holds that every reflex vertex $r\in R\setminus\{a,b\}$ sees a subsegment $I(r)\subseteq s$.
For convenience, we also call~$I(r)$ an \emph{interval}.

\begin{figure}
    \centering
    \includegraphics{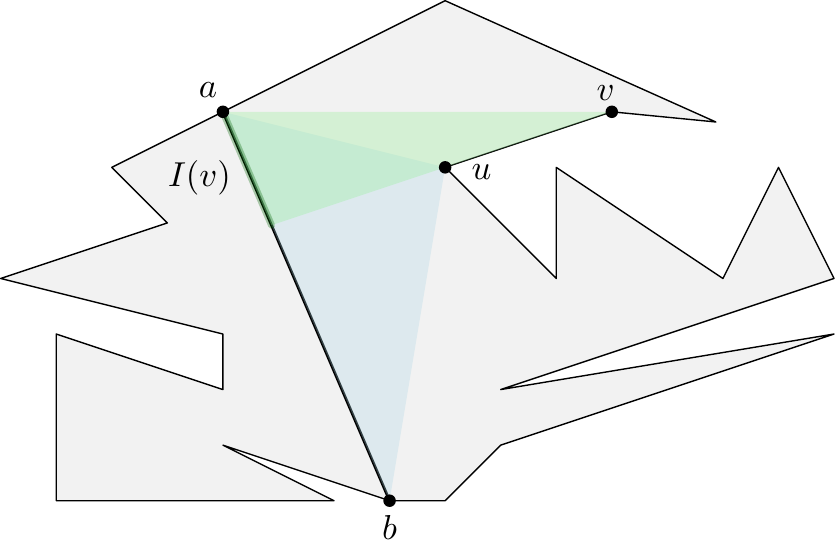}
    \caption{The vertex $v$ sees a subinterval~$I(v)\subseteq s$ which is restricted by $a$ and $u$.}
    \label{fig:sub-interval}
\end{figure}

Note that every interval is \emph{restricted} by exactly two points in $R$, see Figure~\ref{fig:sub-interval}.
In case of ambiguity, due to collinearities, we say the point in $R$ closer to $s$ is the restricting point.
Those vertices can be either the endpoints of $s$ ($a$ and $b$)
or a different reflex vertex in $R$.
We show the following claim.
\begin{claim}
    If $u$ is a reflex vertex that restricts the reflex vertex $v$ then it holds that \[I(v) \subseteq I(u).\]
\end{claim}
\begin{proof}
    The triangle $T$ formed by $I(v)$ and $v$ is trivially convex
    and fully contained inside $P$.
    The reflex vertex $u$ is on the boundary of the triangle
    and thus sees every point of $T$.
    In particular also $I(v)$.
\end{proof}

Given the previous claim, we construct the visibility restriction graph $G$ as follows.
The vertices are formed by the points in $R$.
We say that $uv$ forms a directed edge, if $u$ is restricted by $v$.
We summarize a few useful properties of~$G$ in the following claim.
\begin{claim}
    The visibility restriction graph of a polygon with \pointvisibilitywidth at most $k$ has the following properties.
    \begin{enumerate}
        \item 
        \label{itm:TwoSinks}
        The segment endpoints $a,b$ are the only two sinks.
        \item 
        \label{itm:OutDegree}
        The out-degree is two for every vertex $v\in R\setminus \{a,b\}$.
        \item 
        \label{itm:Indegree}
        The in-degree is at most $k-1$ for every vertex $v\in R$.
        \item 
        \label{itm:LongPaths}
        The longest path has at most $k+1$ vertices.
    \end{enumerate}
\end{claim}
\begin{proof}
    By definition, every reflex vertex is restricted by exactly two vertices in $R$. This implies Item~\ref{itm:TwoSinks} and~\ref{itm:OutDegree}.
    
    Any reflex vertex $v$ can see itself and all its neighbors.
    Its in-degree neighbors are also reflex vertices.
    As no point can see more than $k$ reflex vertices $v$ has at most $k-1$ in-degree neighbors.
    This concludes the proof of Item~\ref{itm:Indegree}.
    
    Finally, to prove Item~\ref{itm:LongPaths}, let $p = u_1u_2\ldots u_l$ be a directed path.
    Then it holds that there is a point  
    \[q\in I(u_l) \subseteq \ldots \subseteq I(u_2) \subseteq I(u_1)=s.\]
    The point~$q$ sees all reflex vertices of the path~$p$.
    As no point sees more than $k$ reflex vertices, it holds that $p$ has at most $k$ reflex vertices. 
    As all but potentially the first vertex is a reflex vertex, we have $l\leq k+1$.
\end{proof}

The properties of the last claim enable us to give an upper bound on the size of $G$ and thus also on the size of $R$.
\begin{claim}
    The visibility restriction graph~$G$ of a polygon with \pointvisibilitywidth $\var{pvw} = k$ has at most $k^{O(k)}$ vertices.
\end{claim}
\begin{proof}
    We organize~$G$ into layers depending on the distance from~$a$~and~$b$. 
    Note that if layer~$i$ has~$t$ vertices then layer~$(i+1)$ has at most $t\cdot k$ vertices. As there are at most $k+1$ layers and the first layer has size two we get that $G$ has at most
    \[2 + 2k + 2k^2 + 2k^3 + \ldots + 2k^{k} = k^{O(k)}\] vertices.
\end{proof}

\section{Conclusion}
\label{sec:conclusion}
We believe that local complexity has the potential to be a useful  parameter.
We gave two ways to define local complexity in a rigorous way and 
showed how those two ways relate to one another.
We want to end with a few open questions.
\begin{enumerate}
    \item Can we find algorithms and data structures that can make use of low local complexity? 
    \item Can we compute or approximate the \pointvisibilitywidth and \chordvisibilitywidth in an efficient manner?
    Note that this is more a theoretical question.
    We do not necessarily need to know the \chordvisibilitywidth of a polygon to use the concept in the design and analysis of an algorithm.
    \item Are there other ways to formalize the idea of 
    low local complexity within a polygonal region?
\end{enumerate}

\bibliography{LC-library}
\end{document}